\documentclass{article}
\usepackage{amsthm}
\usepackage{amsmath,amssymb}
\usepackage{tabls}
\usepackage{graphicx}
\usepackage{array}
\usepackage[algoruled,linesnumbered,algo2e,vlined]{algorithm2e}
\usepackage{url}
\usepackage{amssymb}
\usepackage{multirow}
\usepackage{multicol}

\usepackage{a4wide}

\newtheorem{theorem}{Theorem}
\newtheorem{lemma}{Lemma}
\newtheorem{definition}{Definition}

\newcommand{\Suffix}{\mathit{Suffix}}
\newcommand{\OSuffix}{\mathit{Suffix}_{\prec}}

\newcommand{\CST}{\mathit{CST}}
\newcommand{\PH}{\mathit{PH}}
\newcommand{\rev}[1]{#1^\mathit{R}}

\newcommand{\id}{\mathit{id}}

\newcommand{\str}{\mathit{str}}

\newcommand{\SuffixTree}{\mathit{STree}}
\newcommand{\LZTrie}{\mathit{LZ78Trie}}

\date{}

\title{
Constructing LZ78 Tries and Position Heaps in Linear Time for 
Large Alphabets
}
\author{\\
  Yuto Nakashima, Tomohiro I, Shunsuke Inenaga,\\ Hideo Bannai, and
  Masayuki Takeda\\
  {\small Department of Informatics, Kyushu University, Fukuoka
    819-0395, Japan}\\
  {\small {\tt \{yuto.nakashima,inenaga,bannai,takeda\}@inf.kyushu-u.ac.jp}}\\
  {\small {\tt tomohiro.i@cs.tu-dortmund.de}}\\
}

\begin{document}
\maketitle
\begin{abstract}
We present the first 
worst-case linear-time algorithm to compute the Lempel-Ziv 78 factorization
of a given string over an integer alphabet.
Our algorithm is based on nearest marked ancestor queries on the suffix tree 
of the given string.
We also show that the same technique can be used to construct
the position heap of a set of strings in worst-case linear time,
when the set of strings is given as a trie.
\end{abstract}

\section{Introduction}
\emph{Lempel-Ziv 78} (\emph{LZ78}, in short) is a well known compression algorithm~\cite{LZ78}.
LZ78 compresses a given text based on a dynamic dictionary
which is constructed by partitioning the input string,
the process of which is called LZ78 factorization.
Other than its obvious use for compression, the LZ78 factorization
is an important concept used in various string processing
algorithms and
applications~\cite{crochemore03:_subquad_sequen_align_algor_unres_scorin_matric,li05:_lz78_based_strin_kernel}.

In this paper, we show an LZ78 factorization algorithm
which runs in $O(n)$ time using $O(n)$ working space for an integer alphabet,
where $n$ is the length of a given string and $m$ is the size of the LZ78 factorization.
Our algorithm does not make use of any randomization such as hashing,
and works in $O(n)$ time \emph{in the worst case}.
To our knowledge, this is the first $O(n)$-time LZ78 factorization algorithm
when the size of an integer alphabet is $O(n)$
and $2^{\omega(\log n \frac{\log \log \log n}{(\log \log n)^2})}$.
Our algorithm computes the LZ78 trie (a trie representing the LZ78 factors) 
via the suffix tree~\cite{Weiner}
annotated with a semi-dynamic nearest marked ancestor data structure~\cite{westbrook_fast_incre_1992,amir_improved_dynamic_1995}.

We also show that the same idea can be used to construct 
the \emph{position heap}~\cite{ehrenfeucht_position_heaps_2011}
of a set of strings which is given as a trie,
and present an $O(\ell)$-time algorithm to construct it, 
where $\ell$ is the size of the given trie.

Some of the results of this paper appeared 
in the preliminary versions~\cite{position_heaps_of_trie_2012,BannaiIT12}.

\subsection*{Comparison to previous work}
The LZ78 trie (and hence the LZ78 factorization) of a string of length $n$
can be computed in $O(n)$ expected time and $O(m)$ space,
if hashing is used for maintaining the branching nodes of the LZ8 trie~\cite{FialaG89}.
In this paper, we focus on algorithms without randomization, 
and we are interested in the worst-case behavior of LZ78 factorization algorithms.
If balanced binary search trees are used in place of hashing,
then the LZ78 trie can be computed in $O(n \log \sigma)$ worst-case time 
and $O(m)$ working space.
Our $O(n)$-time algorithm is faster than this method
when $\sigma \in \omega(1)$ and $\sigma \in O(n)$.
On the other hand, our algorithm uses $O(n)$ working space,
which can be larger than $O(m)$ when the string is LZ8 compressible.
Jansson et al.~\cite{jansson13:_linked_dynam_tries_applic_lz}
proposed an algorithm which computes the LZ78 trie of a given string 
in $O(n(\log \log n)^2/(\log_{\sigma} n \log \log \log n))$ worst-case time, 
using $O(n(\log\sigma + \log\log_{\sigma} n)/ \log_{\sigma} n)$
bits of working space.
Our $O(n)$-time algorithm is faster than theirs
when $\sigma \in 2^{\omega(\log n \frac{\log \log \log n}{(\log \log n)^2})}$ and $\sigma \in O(n)$,
and is as space-efficient as theirs 
when $\sigma \in \Theta(n)$.
Tamakoshi et al.~\cite{TamakoshiIIBT13} proposed an algorithm 
which computes the LZ78 trie in $O(n + (s+m) \log \sigma)$ worst-case time 
and $O(m)$ working space,
where $s$ is the size of the run length encoding (RLE) of a given string.
Our $O(n)$-time algorithm is faster than theirs
when $\sigma \in 2^{\omega (\frac{n}{s+m})}$ and $\sigma \in O(n)$.

The position heap of a single string of 
length $n$ over an alphabet of size $\sigma$ can be computed in 
$O(n \log \sigma)$ worst-case time and $O(n)$ space~\cite{ehrenfeucht_position_heaps_2011},
if the branches in the position heap are maintained by balanced binary search trees.
Independently of this present work, Gagie et al.~\cite{travis_position_heaps_2013} showed 
that the position heap of a given string of length $n$ over an integer alphabet 
can be computed in $O(n)$ time and $O(n)$ space, via the suffix tree of the string.

\section{Preliminaries}

\subsection{Notations on strings}

We consider a string $w$ of length $n$ 
over integer alphabet $\Sigma = \{1, \ldots, \sigma\}$, where $\sigma \in O(n)$.
The length of $w$ is denoted by $|w|$, namely, $|w| = n$.
The empty string $\varepsilon$ is a string of length 0,
namely, $|\varepsilon| = 0$.
For a string $w = xyz$, $x$, $y$ and $z$ are called
a \emph{prefix}, \emph{substring}, and \emph{suffix} of $w$, respectively.
The set of suffixes of a string $w$ is denoted by $\Suffix(w)$.
The $i$-th character of a string $w$ is denoted by 
$w[i]$ for $1 \leq i \leq n$,
and the substring of a string $w$ that begins at position $i$ and
ends at position $j$ is denoted by $w[i..j]$ for $1 \leq i \leq j \leq n$.
For convenience, let $w[i..j] = \varepsilon$ if $j < i$.
For any string $w$, let $\rev{w}$ denote the reversed string of $w$,
i.e., $\rev{w} = w[n]w[n-1] \cdots w[1]$.

\subsection{Suffix Trees}
We give the definition of a very important and well known string index
structure, the suffix tree.
To assure property~\ref{def:suffixtreeleaf} below for the sake of presentation,
we assume that string $w$ ends with a unique character
that does not occur elsewhere in $w$.
\begin{definition}[Suffix Trees~\cite{Weiner}]
  For any string $w$, its suffix tree, denoted 
  $\SuffixTree(w)$, is a labeled rooted tree which satisfies the following:
  \begin{enumerate}
   \item each edge is labeled with a non-empty substring of $w$;
   \item each internal node has at least two children; 
   \item
     the labels $x$ and $y$ of any two distinct out-going edges from the same node begin
     with different symbols in $\Sigma$;
   \item\label{def:suffixtreeleaf} there is a one-to-one correspondence between 
   the suffixes of $w$ and the leaves of $\SuffixTree(w)$, i.e.,
   every suffix is spelled out by a unique path from the root to a leaf.
  \end{enumerate}
\end{definition}

Since any substring of $w$ is a prefix of some suffix of $w$,
all substrings of $w$ can be represented as a path from the root in $\SuffixTree(w)$.
For any node $v$, let $\str(v)$ denote the string 
which is a concatenation of the edge labels from the root to $v$.
A locus of a substring $x$ of $w$ in $\SuffixTree(w)$ is a pair $(v, \gamma)$ 
of a node $v$ and a (possibly empty) string $\gamma$, 
such that $\str(v)\gamma = x$ and $\gamma$ is the shortest.
A locus is said to be an explicit node if $\gamma = \varepsilon$,
and is said to be an implicit node otherwise.
It is well known that $\SuffixTree(w)$ can be represented with $O(n)$ space,
by representing each edge label $x$ with a pair $(i, j)$ of positions
satisfying $x = w[i..j]$.

\begin{theorem}[\cite{farach97:_optim_suffix_tree_const_large_alphab}]
Given a string $w$ of length $n$ over an integer alphabet,
$\SuffixTree(w)$ can be computed in $O(n)$ time.
\end{theorem}

\subsection{Suffix Trees of multiple strings}

A {\em generalized} suffix tree of a set of strings
is the suffix tree that contains all suffixes of all the strings in the set.
Generalized suffix trees for a set $W = \{w_1\$, \ldots, w_k\$\}$ 
of strings over an integer alphabet
can be constructed in linear time in the total length of the strings.

Suppose that the set $W$ of strings is given as a \emph{reversed trie}
called a common-suffix trie, which is defined as follows.

\begin{definition}[Common-suffix tries~\cite{breslauer_suffix_tree_tree_1998}]
The common-suffix trie of a set $W$ of strings, denoted $\CST(W)$,
is a reversed trie such that 
 \begin{enumerate}
  \item each edge is labeled with a character in $\Sigma$;
  \item any two in-coming edges of any node are labeled with distinct characters;
  \item each node $v$ represents the string
        that is a concatenation of the edge labels in the path from $v$ to the root;
  \item for each string $w \in W$ there exists a unique leaf which represents $w$.
 \end{enumerate}
\end{definition}

\begin{figure}[tb]
\centering{
\includegraphics[scale=0.5]{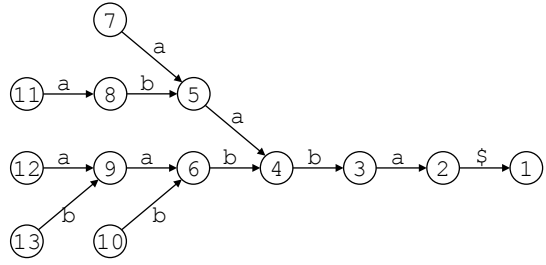}
}
\caption{
  $\CST(W)$ for
  $W = \{\mathtt{aaba\$},$ 
  $\mathtt{bbba\$},$
  $\mathtt{ababa\$},$
  $\mathtt{aabba\$},$
  $\mathtt{babba\$}\}$.
  Each node $u$ is associated with $\id(u)$.}
\label{fig:CST}
\end{figure}

An example of $\CST(W)$ is illustrated in Figure~\ref{fig:CST}.

Let $\ell$ be the number of nodes in $\CST(W)$,
and let $\Suffix(W)$ be the set of suffixes of the strings in $W$,
i.e., $\Suffix(W) = \bigcup_{w \in W} \Suffix(w)$.
Clearly, $\ell$ equals to the cardinality of $\Suffix(W)$ (including the empty string).
Hence, $\CST(W)$ is a natural representation of the set $\Suffix(W)$.
If $L$ is the total length of strings in $W$,
then $\ell \leq L+1$ holds. 
On the other hand, when the strings in $W$ share many suffixes,
then $L = \Theta(\ell^2)$
(e.g., consider the set of strings $\{ab^i \mid 1 \leq i \leq \ell\}$).
Therefore, $\CST(W)$ can be regarded as 
a compact representation of the set $W$ of strings.

\begin{definition}[Suffix Trees for $\CST(W)$]
  For any $\CST(W)$, its suffix tree, denoted $\SuffixTree(W)$,
  is a labeled rooted tree which satisfies the following:
  \begin{enumerate}
   \item each edge is labeled with a non-empty string 
         which is a concatenation of the edge labels of $\CST(W)$;
   \item each internal node has at least two children; 
   \item
     the labels $x$ and $y$ of any two distinct out-going edges from the same node begin
     with different symbols in $\Sigma$;
   \item there is a one-to-one correspondence between 
   the internal nodes of $\CST(W)$ and the leaves of $\SuffixTree(W)$, i.e.,
   every string which is represented by a node in $\CST(W)$ 
   is spelled out by a unique path from the root to a leaf.
  \end{enumerate}
\end{definition}

Notice that the suffix tree for $\CST(W)$ is identical to a generalized suffix tree
of the set $W$ of strings. 
If a given $\CST(W)$ is of size $\ell$, 
then the size of the suffix tree of $\CST(W)$ is $O(\ell)$.

We will use the following result in our algorithms.
\begin{theorem}[\cite{Shibuya_construct_stree_of_tree}]
Given $\CST(W)$ of size $\ell$ for a set $W$ of strings over an integer alphabet,
the generalized suffix tree of $W$ can be computed in $O(\ell)$ time.
\end{theorem}

\subsection{Tools on trees}
We will use the following efficient data structures on rooted trees.

\begin{lemma}[Nearest marked ancestor~\cite{westbrook_fast_incre_1992,amir_improved_dynamic_1995}] \label{lem:nma}
A semi-dynamic rooted tree can be maintained in linear space
so that the following operations are supported in amortized $O(1)$ time:
1) find the nearest marked ancestor of any node;
2) insert an unmarked node;
3)  mark an unmarked node.
\end{lemma}
By ``semi-dynamic" above we mean that no nodes are to be deleted from the tree.

\begin{lemma}[Level ancestor query~\cite{BerkmanV94,BenderF04}] \label{lem:laq}
Given a static rooted tree,
we can preprocess the tree in linear time and space 
so that the $i$th node in the path
from any node to the root can be found in $O(1)$ time
for any integer $i \geq 0$, if such exists.
\end{lemma}

\section{Algorithms}

In this section, we propose algorithms to compute 
\emph{LZ78 factorizations}~\cite{LZ78} and \emph{position heaps}~\cite{ehrenfeucht_position_heaps_2011} which run in linear time for an integer alphabet.

\subsection{Computing LZ78 trie from suffix tree}

The LZ78 factorization~\cite{LZ78} of a string $w$ is a sequence  
$f_1, \ldots, f_m$ of non-empty substrings of $w$,
where $f_1 \cdots f_m = w$, 
and each $f_i$ is the longest prefix of $w[|f_1 \cdots f_{i-1}|+1 \ldots n]$ such that
$f_i \in \{f_jc \mid 1 \leq j < i, c\in\Sigma \}\cup\Sigma$.
Each $f_i$ is called an LZ78 factor of $w$.
The dictionary of LZ78 factors of a string $w$ can be represented by the following trie,
called the \emph{LZ78 trie} of $w$.
\begin{definition}
The LZ78 trie of string $w$, denoted $\LZTrie(w)$,
is a rooted tree such that each node represents an LZ78 factor $f_i$,
and there is an edge $(f_j, c, f_i)$ with label $c \in \Sigma$
iff $f_i = f_j c$.
\end{definition}

\begin{figure}[tb]
\centering{
\includegraphics[scale=0.5]{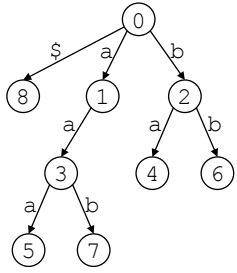}
}
\caption{The LZ78 trie of $\mathtt{abaabaaaabbaab\$}$.
  Each node numbered $i$ represents the factor $f_i$ of the LZ78 factorization,
  where $f_i$ is the path label from the root to the node, 
  e.g.: $f_3 = \mathtt{aa}$, $f_7 = \mathtt{aab}$.
}
\label{fig:lz78dict}
\end{figure}

See Figure~\ref{fig:lz78dict} for an example of $\LZTrie(w)$.
$\LZTrie(w)$ requires $O(m)$ space,
where $m$ is the number of factors in the LZ78 factorization of $w$.
Each factor $f_i$ can be computed in $O(|f_i|)$ time from the trie,
by starting from node $f_i$ and concatenating edge labels between $f_i$ and the root.
We compute $\LZTrie(w)$ as a compact representation of the LZ78 factorization of $w$.

We present an $O(n)$-time algorithm to compute LZ78 trie of string $w$ of length $n$
over an integer alphabet, via the suffix tree of $w$.
In so doing, we make use of the following key observation:
since $\LZTrie(w)$ is a trie whose nodes are all substrings of $w$,
it can be \emph{superimposed} on $\SuffixTree(w)$, and be completely contained in it,
with the exception that some nodes of the trie may correspond to implicit nodes 
of the suffix tree. 
See Figure~\ref{fig:stree_lz} for an example.
\begin{figure}[tbh]
  \centerline{
    \includegraphics[width=0.7\textwidth]{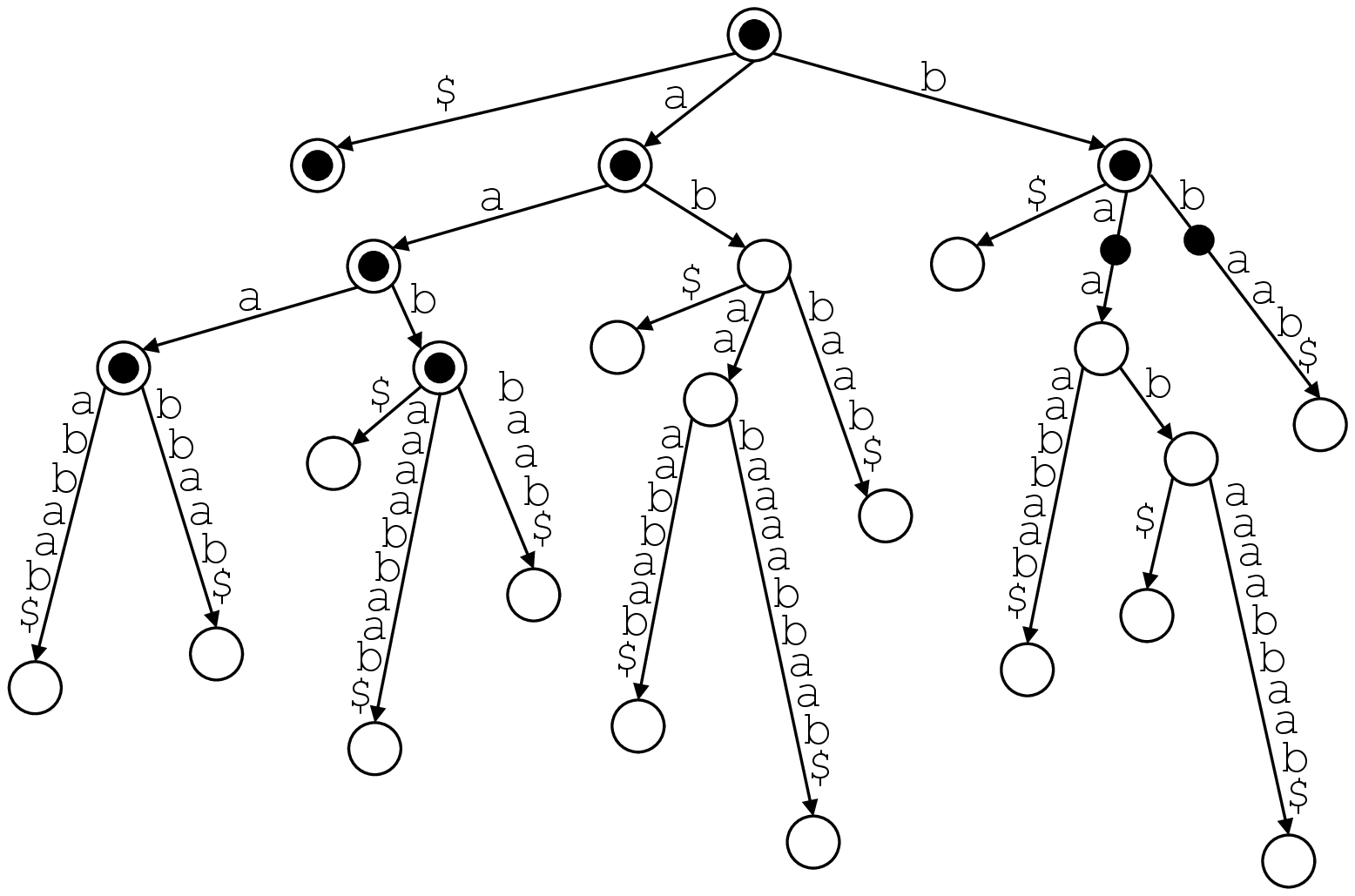}
  }
  \caption{
    The LZ78-trie of string $w = \mathtt{abaabaaaabbaab\$}$,
    superimposed on the suffix tree of $w$.
    The subtree consisting of the dark nodes is the LZ78-trie,
    derived from the LZ78-factorization:
    $\mathtt{a},\mathtt{b},\mathtt{aa},\mathtt{ba},\mathtt{aaa},\mathtt{bb},\mathtt{aab},\mathtt{\$}$,
    of $w$.
  }
  \label{fig:stree_lz}
\end{figure}

\begin{theorem} \label{theo:LZ78trie}
Given a string $w$ of length $n$ over an integer alphabet,
$\LZTrie(w)$ can be constructed in $O(n)$ time and $O(n)$ working space.
\end{theorem}
\begin{proof}
Suppose the LZ78 factorization 
$f_1\cdots f_{i-1}$, up to position $p - 1 = |f_1\cdots f_{i-1}|$ of a given string $w$, has been computed,
and the nodes of the LZ78 trie for $f_1, \ldots, f_{i-1}$ have been added to $\SuffixTree(w)$.
Now, we wish to calculate the $i$th LZ78-factor $f_i$ starting at position $p$.
Let $z$ be the leaf of the suffix tree that corresponds to the suffix $w[p..n]$.
The longest previous factor $x$
that is a prefix of $w[p..n]$ corresponds to the longest path of the LZ78 trie built so far,
which represents a prefix of $w[p..n]$.
If we consider the suffix tree as a semi-dynamic tree, 
where the nodes corresponding to the superimposed LZ78-trie are
dynamically inserted and marked, the node $x$ we are looking for
is the {\em nearest marked ancestor} of $z$,
which can be computed in $O(1)$ time.
If $x$ is not branching,
then we we simply move down the edge by a single character (say $a$),
create a new node if necessary, and mark the node representing 
the $i$th LZ78 factor $f_i = xa$.  
If $x$ is branching,
then we can locate the out-going edge of $x$ that is in the path from $x$ to 
the leaf $z$ in $O(1)$ time by a level ancestor query from $z$.
Then we insert/mark the new node for the $i$th LZ78 factor $f_i$.
Technically, our suffix tree is semi-dynamic in that new
nodes are created since the LZ78-trie is superimposed. 
However, since we are only interested in level ancestor queries at
branching nodes, we only need to answer them for the original suffix
tree. Therefore, we can preprocess the tree in $O(n)$ time and
space to answer the level ancestor queries in $O(1)$ time.
Finally, we obtain the LZ78-trie by removing the unmarked nodes from the provided suffix tree.
\end{proof}

\subsection{Computing position heap from suffix tree}

Here, we show how to compute 
the position heap of a set of strings from the corresponding suffix tree.
We begin with the definition of position heaps.

Let $W = \{w_1 \$, w_2 \$, \ldots, w_k \$\}$ be a set of strings 
such that $w_i \$ \notin \Suffix(w_j \$)$ for any $1 \leq i \neq j \leq k$.
Define the total order $\prec$ on $\Sigma^*$ by 
$x \prec y$ iff either $|x| < |y|$ 
or $|x| = |y|$ and $\rev{x}$ is lexicographically smaller than $\rev{y}$.
Let $\OSuffix(W)$ be the sequence of strings 
in $\Suffix(W)$ that are ordered w.r.t. $\prec$
and let $\ell = |\OSuffix(W)|$.
For any $1 \leq i \leq \ell$,
let $s_i$ denote the $i$th suffix of $\OSuffix(W)$.

\begin{definition}[Position heaps for multiple strings] \label{def:position_heap}
The \emph{position heap} for a set $W$ of strings, 
denoted $\PH(W)$, is the trie heap defined as follows:
\begin{enumerate}
  \item each edge is labeled with a character in $\Sigma$;
  \item any two out-going edges of any node are labeled with distinct characters;
  \item the root node is labeled with $1$ and the other nodes are labeled with
        integers from $1$ to $\ell$ such that parents' labels are smaller than 
        their children's labels;
  \item the path from the root to node labeled with $i~(1 \leq i \leq \ell)$ 
        is a prefix of $s_i$.
\end{enumerate}
\end{definition}

\begin{figure}[tb]
\centering{
\includegraphics[scale=0.5]{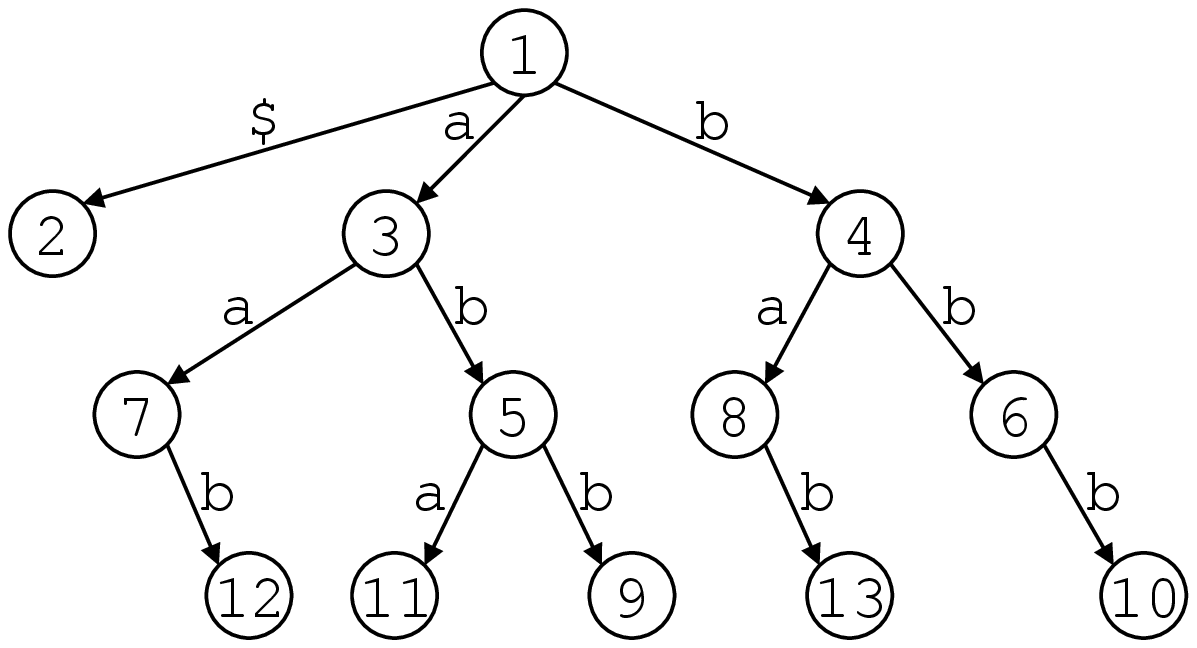}
}
\caption{$\PH(W)$ for
  $W = \{\mathtt{aaba\$},\mathtt{bbba\$},\mathtt{ababa\$},\mathtt{aabba\$},\mathtt{babba\$}\}$,
  where $\Suffix_{\prec}(W) =$ $\langle \varepsilon,$
  $\mathtt{\$},$
  $\mathtt{a\$},$
  $\mathtt{ba\$},$ 
  $\mathtt{aba\$},$
  $\mathtt{bba\$},$
  $\mathtt{aaba\$},$
  $\mathtt{baba\$},$
  $\mathtt{abba\$},$
  $\mathtt{bbba\$},$
  $\mathtt{ababa\$},$
  $\mathtt{aabba\$},$
  $\mathtt{babba\$} \rangle$.
The node labeled with integer $i$ corresponds to $s_i$.}
\label{fig:PH}
\end{figure}

Notice that $\PH(W)$ can be obtained 
by inserting, into a trie, the strings in $\OSuffix(W)$ in increasing order w.r.t. $\prec$.
In this paper, we assume that $\CST(W)$ is given as input,
but $\OSuffix(W)$ is not explicitly given.
For each $s_i$ in $\OSuffix(W)$, let $\id(s_i) = i$.
We would like to know $\id(s)$ for all suffixes $s$ represented by $\CST(W)$,
which gives us the ordering of strings in $\CST(W)$ w.r.t. $\OSuffix(W)$.
\begin{lemma}
$\id(s)$ for all nodes $s$ in $\CST(W)$ can be computed in $O(\ell)$ time.
\end{lemma}
\begin{proof}
We firstly construct the suffix array of $\CST(W)$ in $O(\ell)$ time,
using the algorithm proposed by Ferragina et al.~\cite{FerraginaLMM09}.
This gives us the lexicographical order of the suffixes represented by $\CST(W)$.
Secondly, we bucket-sort the nodes of $\CST(W)$:
we use an array of size $x$ as buckets, 
where $x \leq \ell$ is the length of the longest string in $\CST(W)$.
We then scan the suffix array from the beginning to the end, 
and insert each node (string) $s$ into the $|s|$th bucket (entry) of the array. 
This gives us $\id(s)$ for all nodes $s$ in $\CST(W)$ in $O(\ell)$ time.
\end{proof}

For any $1 \leq i \leq \ell$, where $\ell$ is the number of nodes of $\CST(W)$,
let $\CST(W)^{i}$ denote the subtree of $\CST(W)$ 
consisting only of the nodes $s_j$ with $1 \leq j \leq i$.
$\PH(W)^{i}$ is the position heap for $\CST(W)^i$ for each $1 \leq i \leq \ell$,
and in our algorithm which follows, 
we construct $\PH(W)$ incrementally, in increasing order of $i$.

We present an $O(\ell)$-time algorithm to compute $\PH(W)$ from 
the generalized suffix tree for $\CST(W)$ with $\ell$ nodes.
Since $\PH(W)$ is a trie where each node represents some substring of the strings in $W$,
it can be \emph{superimposed} on the generalized suffix tree of $W$
which is equivalent to the suffix tree of $\CST(W)$,
and be completely contained in it,
with the exception that some nodes of the trie may correspond to implicit nodes 
of the suffix tree. 
See Figure~\ref{fig:stree_ph} for an example of $\PH(W)$ superimposed to
the suffix tree of $\CST(W)$.
We summarize our algorithm as follows.

\begin{figure}[tb]
  \centerline{
    \includegraphics[width=0.7\textwidth]{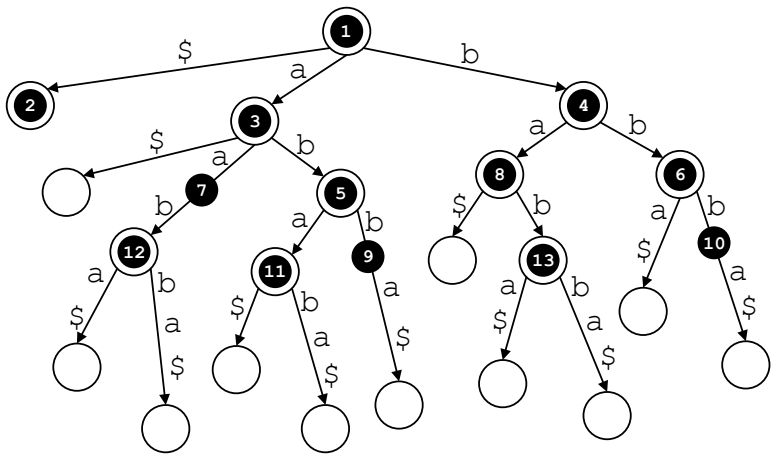}
  }
  \caption{
    The position heap of set $W = \{\mathtt{aaba\$},\mathtt{bbba\$},\mathtt{ababa\$},\mathtt{aabba\$},\mathtt{babba\$}\}$,
    superimposed on the generalized suffix tree of $W$,
    which is equivalent to the suffix tree of $\CST(W)$.
    The subtree consisting of the dark nodes is the position heap of $W$.
  }
  \label{fig:stree_ph}
\end{figure}

\begin{theorem}
Given $\CST(W)$ with $\ell$ nodes representing a set $W$ of strings over an integer alphabet,
$\PH(W)$ can be constructed in $O(\ell)$ time and $O(\ell)$ working space.
\end{theorem}
\begin{proof}
Suppose we have computed the position heap $\PH(W)^{i-1}$
superimposed onto the suffix tree of $\CST(W)$,
and we wish to find the next node which corresponds to suffix $s_i$.
Let $z$ be the leaf of the suffix tree that corresponds to the suffix $s_i$.
The longest prefix of $s_i$ that is represented by $\PH(W)^{i-1}$ corresponds to
the longest path of $\PH(W)^{i-1}$, which represents a prefix of $s_i$.
Therefore, this can be found by a semi-dynamic nearest marked ancestor query,
and the rest is analogous to the algorithm to compute the LZ78 trie
of Theorem~\ref{theo:LZ78trie}.
Finally, we obtain the position heap by removing 
the unmarked nodes from the provided suffix tree.
\end{proof}

\bibliographystyle{plain}
\bibliography{ref}
\end{document}